\documentclass[runningheads]{llncs}
\usepackage{graphicx}
\usepackage{amsmath} 
\usepackage{amssymb, amsfonts, mathrsfs,dsfont} 

\usepackage{color,psfrag}

\usepackage{pgf, tikz}
\usetikzlibrary{arrows, automata}

\newcommand{\until}[1]{\{1,\dots, #1\}}

\newcommand{\real}{\mathbb{R}}

\newcommand{\be}{\begin{equation}}
\newcommand{\ee}{\end{equation}}
\newcommand{\ba}{\begin{array}}
\newcommand{\ea}{\end{array}}

\renewcommand{\d}{\textup{d}}
\newcommand{\cl}[1]{\textup{cl}[{#1}]}

\begin{document}
\title{Modeling limited attention in opinion dynamics by topological interactions\thanks{Supported in part by MITI CNRS via 80 PRIME grant DOOM and by ANR via project HANDY, number ANR-18-CE40-0010.}}
%
%
\author{
Francesca Ceragioli\inst{1}
\and
Paolo Frasca\inst{2}
\and
Wilbert Samuel Rossi\inst{3}
}
\authorrunning{F. Ceragioli, P. Frasca, W.S. Rossi}
%
\institute{
DISMA, Politecnico di Torino, Turin, Italy\\
\email{francesca.ceragioli@polito.it}
\and
Univ.\ Grenoble Alpes, CNRS, Inria, Grenoble INP, GIPSA-lab, 38000 Grenoble, France\\
\email{paolo.frasca@gipsa-lab.fr}\\
\and
University College Groningen, University of Groningen, The Netherlands\\
\email{w.s.rossi@rug.nl}}
\maketitle              
\begin{abstract}
This work explores models of opinion dynamics with opinion-dependent connectivity. Our starting point is that individuals have limited capabilities to engage in interactions with their peers. Motivated by this observation, we propose a continuous-time opinion dynamics model such that  interactions take place with a limited number of peers: we refer to these interactions as {\em topological}, as opposed to {\em metric} interactions that are postulated in classical bounded-confidence models. We observe that topological interactions produce equilibria that are very robust to disruptions.
\end{abstract}

\section{Introduction}

Driven by the evolution of digital communication and social networking services, there is an increasing interest for mathematical models of opinion dynamics in social networks. Among the many models proposed in the literature, a few have become popular in the control community
\cite{PT:2017,PT:2018:part2}. In the perspective of the control community, opinion dynamics distinguish themselves from consensus dynamics because consensus is prevented by some other feature of the dynamics. In many popular models, this feature is an opinion dependent limitation of the connectivity. Chief examples are bounded confidence models~\cite{krause:2000:discrete,Deffuant:2000}, where social agents influence each other iff their opinions are closer than a threshold.

This way of defining influence assumes that agents have always access to the opinions of all fellow agents and may lead to agents being influenced by a large number of their fellows, possibly the whole population. Instead, the number of possible interactions is bounded in  practice by the limited time and efforts that individuals can devote to social interactions. Similar {\em limitations of attention} are well documented in psychology and sociology, for instance by the notion of Dunbar number~\cite{DUNBAR1992469}, and become evermore crucial in today's age of information bonanza~\cite{gonccalves2011modeling}. Indeed, in online social media, natural limitations of attention interplay with the way the online platforms are designed.
Users interact via the contents they share: out of the pool of fresh contents, the online platform selects for each user the best contents in order to maximize engagement, mainly based on similarities between users~\cite{lazer2015rise}.
As the notion of Dunbar number was originally defined with reference to primates, the reader will not find surprising that similar ideas have also been fundamental in the study of flocking in animal groups, as testified by numerous theoretical and experimental works
\cite{Ballerini:2008:evidence,Giardina:2008:collective,Frasca:2011:animal-anisotropic,SM:14,Frasca:2017:birds-wires}.
The importance of considering networks where the number of neighbors is limited has also been understood by graph theorists, who have studied the properties of what they call $k$-nearest-neighbors graph: for instance, it is known that $k$ must be logarithmic in $n$ to ensure connectivity~\cite{Balister:2005:connectivity}.

However, few works have incorporated this important observation in suitable models of opinion dynamics. Before surveying these important references, we briefly describe the contribution of  this paper.
In our effort to make the case for limited attention in opinion dynamics, we study a simple continuous-time model (first appeared in the survey paper \cite{Piccoli:2017:chapter}) in which every agent is influenced by her closest  $k$ nearest neighbors. 
In this paper, we provide some preliminary results about this continuous-time dynamics. Our results concentrate on two axes: studying the main properties of its equilibria, including their robustness to disruptions, and proving convergence results in special cases.  We describe the equilibria of the dynamics, distinguishing a special type of {\em clusterization} equilibria that are constituted of separate clusters, and we discuss the robustness of clustered equilibria to disruptions, such as the addition of new agents.
Regarding the question of convergence, we are able to provide a proof in two cases: when the total number of agents $n$ is small enough compared to number of neighbors $k$, namely $n\leq2k+1$,
and when $k=1$, that is, agents are only influenced by one ``best friend''.
 
The difficulties in studying $k$-nearest-neighbors dynamics originate from two key features: (1) interactions are not reciprocal; (2) whether two agents interact does not only depend on their two states, but also on the states of all the other agents.
In the literature, models with any of these features are still relatively few. In classical bounded confidence models~\cite{krause:2000:discrete,Deffuant:2000}, interactions are reciprocal as long as the interaction thresholds are equal for all agents~\cite{krause:2000:discrete,VDB-JMH-JNT:09,VDB-JMH-JNT:09a,FC-PF:11,CC-FF-PT:12},
and any lack of reciprocity makes the analysis much more delicate~\cite{AM-FB:11f,BC-CW:17,GC-WS-WM-FB:19}. In our model, not only interactions are non-reciprocal, but they are also non-metric: whether two agents interact is not solely determined by the distance between their two opinions. For this reason, we follow a consolidated tradition~\cite{Ballerini:2008:evidence} and refer to our connectivity model where agents can interact with their $k$ 
 nearest neighbors as {\em topological}.

Topological interactions are becoming increasingly popular in the applied mathematics community, especially for second order models~\cite{Guo:2017:consensus-flocks}. Kinetic and continuum models with topological interactions are also actively studied~\cite{AB-PD:16,RS-ET:18,PD-MP:19}.
Among first order ``opinion'' models, \cite{DA-SM:19} has recently used Petri nets to define a class of models where interactions depend on the opinions of multiple agents: despite some similarities, our model does not belong to this class. 
In our recent papers~\cite{WSR-PF:20,WSR-PF:2018}, we have studied two dynamics with asynchronous updates (with and without sub-sampling, respectively) that are discrete-time counterparts of the model we propose here. 
Finally, our contribution here differs from the one of~\cite{Piccoli:2017:chapter} as the latter focuses on specific properties of the equilibria, such as the distribution of their clusters' sizes, studied by extensive simulations, whereas we are interested in analytical results about dynamical properties like convergence to the equilibria and about their robustness to perturbations. Our robustness analysis is inspired by the approach taken in~\cite{VDB-JMH-JNT:09} for bounded confidence models.

The rest of this paper has the following structure. Section~\ref{sect:model} introduces the model, Section~\ref{sect:anal} develops its analysis, and Section~\ref{sect:outro} discusses our results.

\section{Mathematical model}\label{sect:model}

Let $n$ and $k$ be two integers with
$$1 \leq k < n,$$ 
and let $V = \{1,\ldots,n\}$ be the set of agents.  Each agent is endowed with a scalar opinion $x_i \in \real$.
For every agent $i\in V$, her neighborhood $N_i$ is defined in the following way. The elements of $V\setminus\{i\}$ are ordered by increasing values of $|x_j-x_i|$; then, the first $k$ elements of the list (i.e. those with smallest distance from $i$) form the set $N_i$ of current neighbors of $i$.
Should a tie between two or more agents arise, priority is given to agents with lower index. 
Note that $N_i$ depends on the state, namely one should write $N_i(x(t))$: nevertheless, we omit to explicitly write the dependence of $N_i$ on the state.
Based on the current definition of $N_i$, agent $i$'s opinion evolves according to
\begin{equation}\label{eq:model}
\dot x_i=\sum_{\ell \in N_i} (x_\ell-x_i)
\end{equation}
We denote by $F(x)$ the righthand side of \eqref{eq:model}. 
In order to describe the inter-agent interactions allowed by a state $x \in \real^n$, it is convenient to define the directed graph 
$$G(x) = (V, E(x)) \quad \text{with}\quad E(x) = \bigcup_{i\in V}  \{ (i,j), j\in N_i\}\,, $$
where $N_i$ is the set of neighbors of $i$. Clearly, if $k=n-1$ the graph $G(x)$ is complete (up to self-loops). In using some simple graph theory in this paper, we take some background and standard jargon for granted: a concise summary can be found in \cite[Ch.~1]{FF-PF:17}. 
The chosen tie-breaking rule makes the right-hand side $F(x)$ well defined for any $x\in\real^n$. The neighborhoods depend on the current state and, therefore, on time. This fact makes dynamics \eqref{eq:model} a piecewise-continuous system~\cite{Ceragioli2018}. 
Its solutions shall be intended in a {\em semi-classical} sense, that is, as piecewise-smooth solutions $\phi(t)$ such that the right-derivative of the solution is equal to the right-hand side at all times, that is, $$\lim_{h\to0^+}\frac{\phi(t+h)-\phi(t)}{h}=F(\phi(t)) \quad \text{for all times $t$}.$$  We conjecture that a forward complete and unique solution exists from every initial condition: a rigorous verification of this fact, which is assumed to hold true in the rest of this paper, is left to future work. Note that choosing a more general notion of solutions, e.g. Caratheodory's, would prevent unicity and the produced multiple solutions would make the subsequent analysis more delicate.  


%
%
%
%
%
%

\section{Analysis}\label{sect:anal}
This section details our results on dynamics~\eqref{eq:model}. We first study equilibria, then convergence properties, and finally reconsider equilibria to study their stability.


\subsection{Equilibria}

A {\em cluster} is a subset of agents that have the same opinion: $C\subset V$ such that $x_i=x_j$ for all $i,j\in C$.
A state $x$ is called \emph{clusterization} if every agent belongs to a cluster of at least $k+1$ elements.
Finally, a clusterization with only one cluster is said to be a \textit{consensus}.

A state $x \in \real^n$ is said to be an \emph{equilibrium} for~\eqref{eq:model} when the right-hand side $F(x)$ is zero.
For any $i\in \until{n}$, it is immediate to see that $\dot x_i=0$ if $i$ belongs to a cluster of at least $k+1$ elements.
This condition is also necessary when $i$ is the index of the smallest or of the largest component.
Therefore, all clusterizations are equilibria and all non-consensus equilibria have at least two clusters of at least $k+1$ elements, but not all equilibria are clusterizations. It is possible to obtain a simple counterexample by considering $k=2$ and $n=7$ with 
\begin{align}
\begin{split}\label{eq:example-k2n7}
x_1=x_3=x_5=0\,, \quad
x_7 = \frac12\,, \quad
x_2 = x_4 = x_6=1.
\end{split}
\end{align}
Note that this example exploits the tie-breaking rule. However, this is not necessary, as the following example shows: 
consider $k=4$ and $n=14$ with 
\begin{align*}
&x_1=x_2=x_3=x_4=x_5=0\,, \\
&x_6=x_7 = \frac25\,, \quad
x_8=x_9 = \frac35\,,\\ 
&x_{10} = x_{11} = x_{12}=x_{13}=x_{14}=1.
\end{align*}

\subsection{Dynamical properties}
We can readily observe that, for any two agents $i$ and $j$, 
\begin{equation} \frac{\d }{\d t} (x_i-x_j)=\sum_{\ell \in N_i\setminus N_j} (x_\ell-x_i) - \sum_{m \in N_j\setminus N_i} (x_m-x_j) - |N_i\cap N_j|\, (x_i-x_j).\label{eq:basic-observation}\end{equation}
This formula allows us to derive a few  consequences. First, we observe that if $N_i(t)=N_j(t)$ for all $t\ge t_0$, then $x_i-x_j\to 0$. Second, we can deduce that the dynamics preserves the order of the agents.

\begin{proposition}[Order preservation]\label{prop:order}
If $x_{i}(t_0)>x_{j}(t_0),$ then $x_{i}(t)>x_{j}(t)$ for all $t\ge t_0$. 
\end{proposition}
\begin{proof}
Observe that \eqref{eq:basic-observation} can be rewritten as
 \begin{align*}
 \frac{\d }{\d t} (x_i-x_j)
 &=\sum_{\ell \in N_{i}\setminus N_{j}} x_\ell - \sum_{\ell \in N_{j}\setminus N_{i}} x_m - k (x_{i}-x_{j})\\&\ge-k(x_{i}-x_{j}),
 \end{align*}
where the inequality holds because $|N_{i}\setminus N_{j}|=|N_{j}\setminus N_{i}|$ and $x_\ell \ge x_m$ for any $\ell \in N_{i}\setminus N_{j}$ and $m\in N_{j}\setminus N_{i}$. 
By this bound and Gronwall lemma, $x_{i}-x_{j}$ cannot reach zero in finite time.
\qed\end{proof}
As a consequence of this property, we can assume from now on with no loss of generality that the agents are {\em sorted in ascending order} of opinions, that is, $x_i(t)\le x_{i+1}(t)$ for all $i\in\until{n-1}$ and all $t\ge0$. The following proposition formally justifies this fact.

\begin{proposition}[Re-ordering agents]
Let $x(t)$ be a solution and $\sigma$ be a permutation on the index set $\until{n}$. Assume that for all pairs of distinct indices $i,j$ the permutation satisfies $\sigma(i)<\sigma(j)$ if either $x_{i}(0)< x_{j}(0)$ or $x_{i}(0)=x_{j}(0)$ and $i< j$. 
Then, the following facts hold true: 
\begin{enumerate}
\item $x_{\sigma(i)}(t)=x_i(t)$ for all $i\in\until{n}$ and for all $t\ge0$;
\item if $\sigma(i)<\sigma(j)$, then $x_{\sigma(i)}(t)\le x_{\sigma(j)}(t)$ for all $t\ge 0$.
\end{enumerate}
\end{proposition}
\begin{proof}
To verify the first claim, notice that the definition of $\sigma$ does not interfere with the tie-breaking rule that is used in the definition of the neighborhoods, therefore $N_{\sigma(i)}=N_i$ and the dynamics of the agent that before the permutation had index $i$ is unchanged.

To verify the second claim, observe the following facts.
If  $x_{\sigma(i)}(0)< x_{\sigma(j)}(0)$, then $x_{\sigma(i)}(t)< x_{\sigma(j)}(t)$ for $t>0$ by Proposition~\ref{prop:order}.
If $x_{\sigma(i)}(0)= x_{\sigma(j)}(0)$, then $N_{\sigma(i)}= N_{\sigma(j)}$ and therefore $x_{\sigma(i)}(t)= x_{\sigma(j)}(t)$ also for $t>0$ by \eqref{eq:basic-observation}.\qed
\end{proof}

%
%
%
%
%
%

From now on we will assume that agents are sorted in ascending order.
We  can now deduce a convergence result for small groups.
\begin{proposition}[Consensus for small groups]\label{prop:small}
If $n\le 2k+1$, then $x(t)$ converges to a consensus. 
\end{proposition}
\begin{proof}
Since $n\le 2k+1$, the two agents with lowest and highest opinion share at least one neighbor. 
Therefore, their difference evolves according to
\begin{align*} \frac{\d }{\d t} (x_{n}-x_{1})&=\!\!\!\!\!\!\sum_{\ell \in N_{n}\setminus N_{1}} \!\!\!\!\!\!(x_\ell-x_{n}) -\!\!\!\!\!\! \sum_{\ell \in N_{1}\setminus N_{n}} \!\!\!\!\!\!(x_\ell-x_{1}) - |N_{1}\cap N_{n}|\, (x_{n}-x_{1})\\&\le-(x_{n}-x_{1}),\end{align*}
which implies exponential convergence to zero by Gronwall lemma. \qed
\end{proof}

Simulations suggest that the dynamics converge also for larger groups, though not necessarily to consensus; see Figure~\ref{fig:simulo}.

\begin{figure}
\psfrag{Time}{$t$}
\psfrag{States}{$x$}
\psfrag{Title}{}
\includegraphics[width=0.48\textwidth]{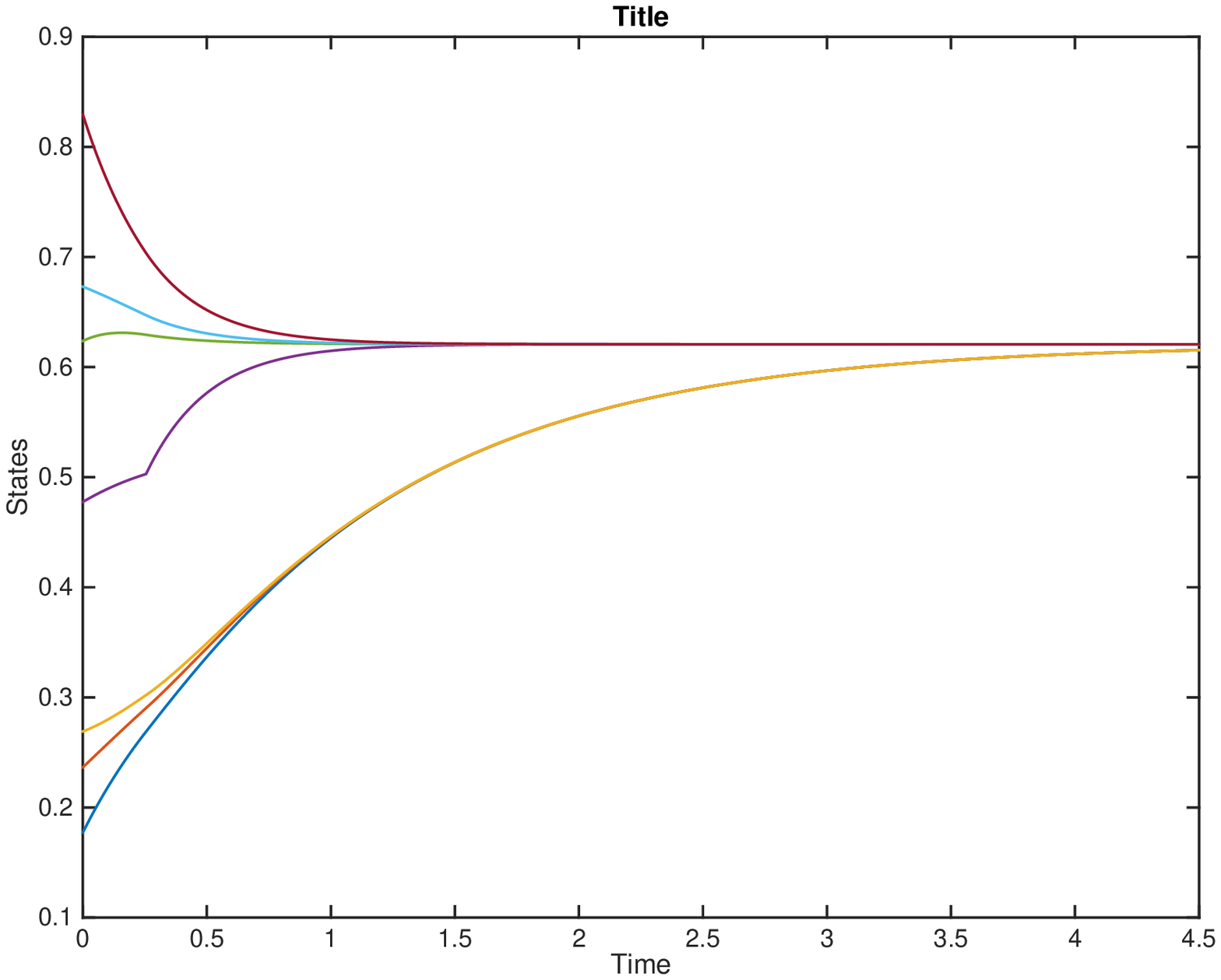}\qquad
\includegraphics[width=0.48\textwidth]{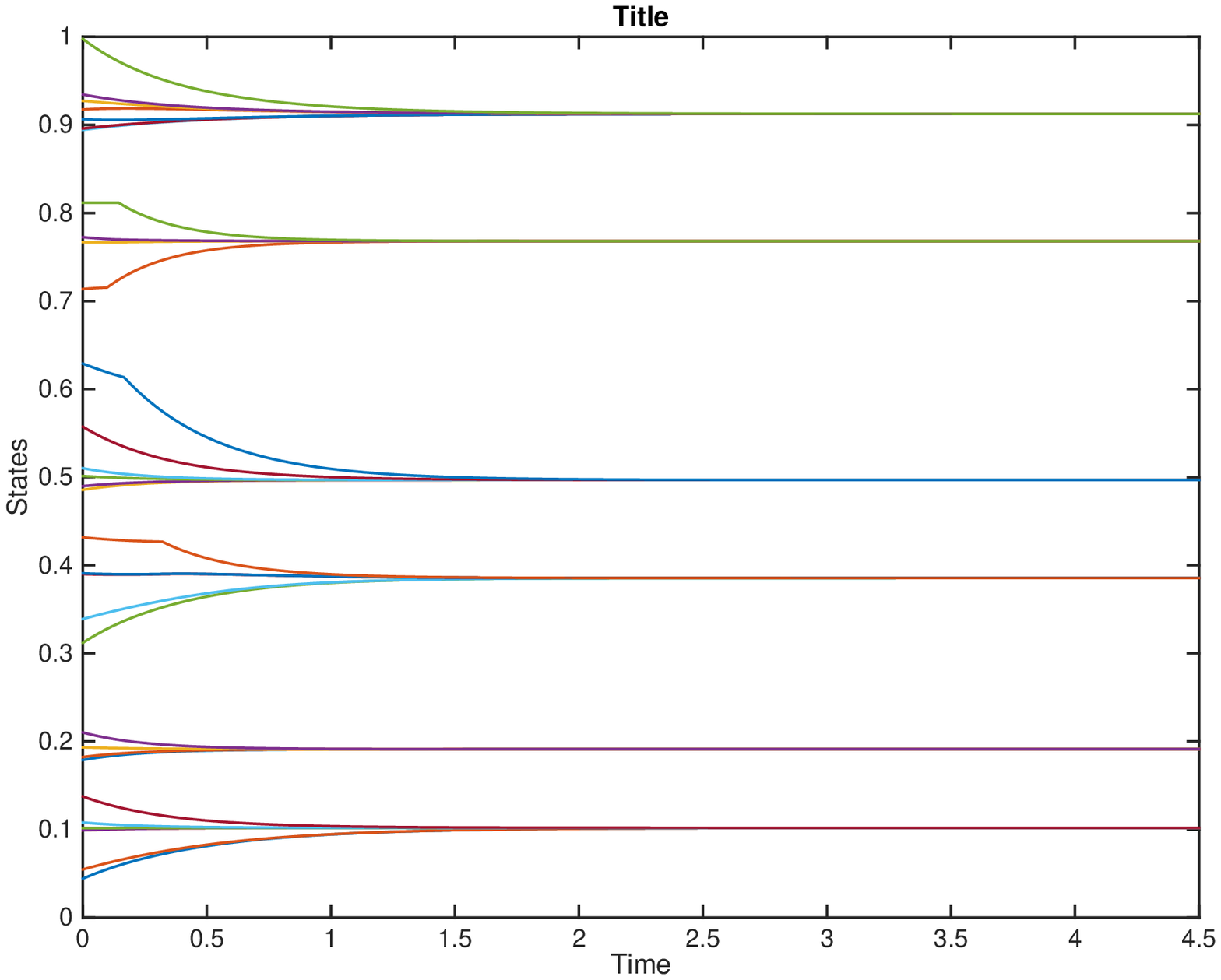}
\caption{Two typical evolutions of the dynamics with $k=3$ from random initial conditions. We observe convergence to consensus for $n=7$ (left) and to a clusterization for $n=30$ (right). The non-smooth nature of the trajectories is also very visible.}\label{fig:simulo}
\end{figure}

\subsection{Special case $k=1$}
In the case $k=1$, the dynamics takes the form $$ \dot x_i= x_{\cl{i}}-x_{i} \quad i\in\until{n},$$
where ${\cl{i}}$ denotes the index of the closest agent to $i$.
This specific form has three important consequences.

\begin{lemma}\label{lemma:case_k=1}
If $k=1$, then the following facts hold true.
\begin{enumerate}
\item All equilibria are clusterizations. 

\item For every $x\in\real^n$, the interaction graph $G(x)$ is the union of weakly connected components, such that each component contains exactly one circuit of length 2 and the two nodes of the circuit can be reached from all nodes of the component.

\item Two disconnected components cannot become connected in the evolution.

\end{enumerate}
\end{lemma}
\begin{proof} 
Claim~1: We observe that the only possibility for the right-hand side to be zero is that $x_{\cl{i}}=x_{i}$ for all $i$. 

Claim~2:  Observe that $\cl{i}$ can only be equal to either $i-1$ or $i+1$, except for the extreme agents, for which necessarily $\cl{1}=2$ and $\cl{n}=n-1$. Therefore, the sequence $\delta_i=\cl{i}-i$ is such that $\delta_1=1$ and $\delta_{n}=-1$ and must therefore change sign an odd number of times. Where it changes from positive to negative, there is a pair of reciprocal edges; where it changes from negative to positive, there is a disconnection. Therefore, every connected component has a pair of nodes that are connected to each other and that can be reached through a directed path from all other nodes. See Figure~\ref{fig:weak-graph} for an illustrative example.

Claim~3: Let there be a disconnection between $j$  and $j+1$. Then, 
$$\frac{\d }{\d t} (x_{j+1}-x_{j})= (x_{j+2}-x_{j+1}) + (x_{j}-x_{j-1})\ge 0,$$
implying that the distance $x_{j+1}-x_{j}$ cannot decrease. Moreover,
$$\frac{\d }{\d t} (x_{j}-x_{j-1})= - (x_{j}-x_{j-1})- (x_{\cl{j-1}}-x_{j-1})\le 0,$$
because the second term either is negative or, if positive, must be smaller or equal in magnitude than $x_{j}-x_{j-1}$.
Therefore, $x_{j}-x_{j-1}$ cannot increase. Since an analogous reasoning implies that $x_{j+2}-x_{j+1}$ cannot increase, the two components cannot become connected in the future. \qed
\end{proof}
\begin{figure}
\begin{center}
    \begin{tikzpicture}[
            > = stealth, 
            shorten > = 0pt, 
            auto,
            node distance = 22mm, 
            semithick 
        ]

        \tikzstyle{every state}=[
            draw = black,
            thick,
            fill = white,
            minimum size = 4mm
        ]

%
%
        
        \node[state] (1) {$1$};
        \node[state] [right of=1] (2) {$2$};
        \node[state] [right of=2] (3) {$3$};                
        \node[state] [right of=3] (4) {$4$};
        \node[state] [right of=4] (5) {$5$};        
        \node[state] [right of=5] (6) {$6$};                
        
        \path[->] (1) edge (2);
        \path[->] (2) edge (3);
        \path[->] (3) edge (4);
        \path[->, bend left] (4) edge (5);   
        \path[->, bend left] (5) edge (4);
        \path[->] (6) edge (5);
        
        \end{tikzpicture}
        \end{center}
\caption{Example of weakly connected component of graph $G(x)$.}\label{fig:weak-graph}
\end{figure}
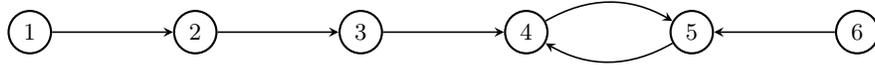

These facts allow to draw a conclusion about convergence.
\begin{proposition}[Clusterization]\label{prop:k1}
If $k=1$, then $x(t)$ converges to a clusterization.
\end{proposition}
\begin{proof}
The third statement of Lemma~\ref{lemma:case_k=1} implies that weakly connected component can only split.
Since the number of individuals is finite, the splitting process terminates with a finite number of constant weakly connected components.
After that termination time, the topology does not change. Since each connected component has a globally reachable node, then each group of agents is guaranteed to converge to consensus~\cite[p.~61]{FF-PF:17}, therefore producing a clusterization. \qed
\end{proof}

Unfortunately, the idea of the proof of Proposition~\ref{prop:k1} does not extend to $k>1$, because in general disconnected components can become connected.

%

\subsection{Stability and robustness of equilibria}

It is easy to see that non-clusterization equilibria are not stable in general.  For instance, consider example \eqref{eq:example-k2n7} with a small perturbation on agent 7: the ensuing dynamics leads to a clusterization with two clusters.
Instead, clusterizations exhibit several stability properties.  

We shall begin by considering small perturbations of the opinions. We say that a clusterization state is {\em structurally stable} if, after a perturbation, the dynamics  converges to another clusterization that has the same clusters (though not necessarily taking on the same opinion values). More formally, the clusterization $\bar x$ is said to be structurally stable if there exists a neighborhood of $\bar x$ such that, for every $y'$ in that neighborhood, any solution issuing from $y'$ converges to a clusterization $\bar y$ that has the following property: for any pair $i,j$ of individuals, $\bar x_i=\bar x_j$ if and only if $\bar y_i= \bar y_j$.
\begin{proposition}[Structural stability of small clusters]
A clusterization is structurally stable if and only if all of its clusters have cardinality not larger than $2k+1$.
\end{proposition}
\begin{proof}
If all clusters have cardinality not larger than $2k+1$, then after the perturbation Proposition~\ref{prop:small} can be applied. To prove the opposite implication, observe that if one cluster has cardinality at least $2k+2$, then a suitable perturbation can split it into two separate clusters of cardinality at least $k+1$, thereby creating a clusterization with different structure.\qed
\end{proof}

We shall also consider different kinds of disruptions, namely the addition or removal of one agent.  We say that a clusterization is stable to these disruptions if, after the addition or removal of an agent, the other agents do not change their opinion.

\begin{proposition}[Stability to removals]
A clusterization is stable to removals if and only if all of its clusters have cardinality larger than $k+1$. 
\end{proposition}
\begin{proof}It is clear that agents in a cluster remain at equilibrium after the removal, unless the cluster size goes below the threshold $k+1$. \qed
\end{proof}

\begin{proposition}[Stability to additions]
Every clusterization is stable to additions. 
\end{proposition}
\begin{proof}Agents within a cluster of size at least $k+1$ will not be influenced by any new arrival. \qed
\end{proof}

\section{Conclusion}\label{sect:outro}

The stability properties of the equilibria of dynamics~\eqref{eq:model} should be contrasted with the lack thereof shown by the equilibria of the corresponding metric bounded confidence model, which reads as
\begin{equation}\label{eq:model-metric}
\dot x_i=\sum_{\ell : |x_\ell-x_i|<d} (x_\ell-x_i),
\end{equation}
where $d>0$ is an interaction radius. It is well-known~\cite{VDB-JMH-JNT:09a,FC-PF:11} that this dynamics converges to clusterizations. If a new agent is added to such a clusterization, either the new agent is too far apart from the original agents and nothing happens, or the new agent falls within the visibility radius from a cluster. In the latter case, the new agents and the agents in the cluster influence each other and therefore change their opinions, converging to a common intermediate value. Actually, if the new agent falls within the visibility radius of two clusters, the two clusters eventually merge. 

In contrast, clusters produced by \eqref{eq:model} are much more stable. In our opinion, this stability intriguingly reminds the stability that is exhibited by norms and organizations in societies. Indeed, sociologists and ethologists have observed since a long time~\cite{GS:98,van2013potent,sethi1996evolutionary} that social norms and social structures are typically rather stable across time, despite the fact that the composition of the social groups evolve, notably with the arrival of new members. Our insights about $k$-neighbor interactions suggest that limitations of attention can have stabilizing effects in societies.

\section*{Acknowledgements}
The authors are grateful to Emiliano Cristiani, Julien Hendrickx, Samuel Martin, Benedetto Piccoli and Tommaso Venturini for fruitful discussions that, along the years, have shaped their point of view on the topic of this paper.

\bibliographystyle{splncs04}


%
\end{document}